%% file: main-arXiv.tex
\documentclass[%
reprint,
superscriptaddress,
nofootinbib,
amsmath,amssymb,
aps,
]{revtex4-2}

\input{pretex}
\input{pretex-extra}

\begin{document}

\title{Symmetry-Based Quantum Circuit Mapping}

\author{Di Yu}
\thanks{D. Y. and K. F. contributed equally to this work.}
\affiliation{Institute for Quantum Computing, Baidu Research, Beijing 100193, China}
\author{Kun Fang}
\email{fangkun02@baidu.com}
\affiliation{Institute for Quantum Computing, Baidu Research, Beijing 100193, China}

\date{\today}

\begin{abstract}
Quantum circuit mapping is a crucial process in the quantum circuit compilation pipeline, facilitating the transformation of a logical quantum circuit into a list of instructions directly executable on a target quantum system.  Recent research has introduced a post-compilation step known as remapping, which seeks to reconfigure the initial circuit mapping to mitigate quantum circuit errors arising from system variability. As quantum processors continue to scale in size, the efficiency of quantum circuit mapping and the overall compilation process has become of paramount importance. In this work, we introduce a quantum circuit remapping algorithm that leverages the intrinsic symmetries in quantum processors, making it well-suited for large-scale quantum systems. This algorithm identifies all topologically equivalent circuit mappings by constraining the search space using symmetries and accelerates the scoring of each mapping using vector computation. Notably, this symmetry-based circuit remapping algorithm exhibits linear scaling with the number of qubits in the target quantum hardware and is proven to be optimal in terms of its time complexity. Moreover, we conduct a comparative analysis against existing methods in the literature, demonstrating the superior performance of our symmetry-based method on state-of-the-art quantum hardware architectures and highlighting the practical utility of our algorithm, particularly for quantum processors with millions of qubits.
\end{abstract}
\maketitle

\begin{figure*}[!htbp]
\includegraphics[width=\textwidth]{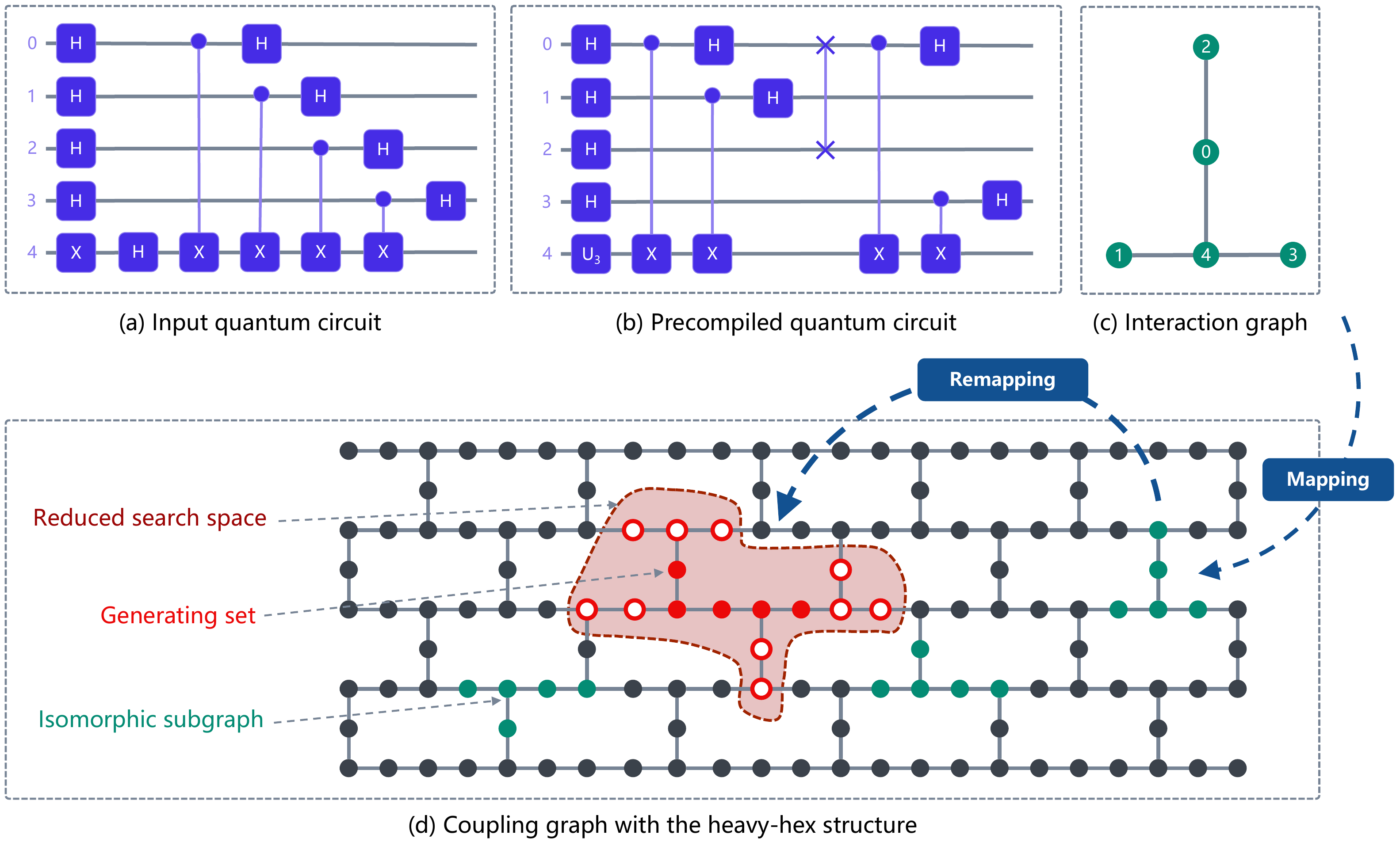}
\caption{\footnotesize An illustration of quantum circuit compilation with mapping and remapping processes. (a) The input quantum circuit. (b) The precompiled quantum circuit, assuming \{U3, CX, SWAP\} as the set of native operations in this example. (c) The interaction graph of the precompiled circuit. (d) A coupling graph with the heavy-hex structure. The mapping process finds a mapping from logical qubits in the interaction graph to physical qubits in the coupling graph. The remapping process finds a mapping from physical qubits to physical qubits such that the composed mapping-and-remapping process reduces the error of circuit implementation. The search for all isomorphic subgraphs employed in the remapping process can be confined to the reduced search space, which is a neighborhood of the generating set, as depicted within the red area.}
\label{fig:circuit_mapping_idea}   
\end{figure*}

\section{Introduction}
\label{section: introduction}

Quantum computing has emerged as a promising avenue for solving intractable problems that are beyond the reach of classical computing, such as integer factoring~\cite{shor1994algorithms}, large database search~\cite{grover1996fast}, chemistry simulation~\cite{lanyon2010towards} and machine learning~\cite{biamonte2017quantum}. To put the theoretically blueprinted quantum advantages into use, it is essential to execute them on real-world quantum computers, moving beyond theoretical concepts and simulation environments. However, quantum algorithms are typically designed at the logical level, assuming ideal qubits and the ability to apply a universal set of quantum gates. These algorithms are not directly executable on quantum hardware devices, which comprise qubits with limited coherence time and support only a specific set of native operations. This disparity underscores the necessity of quantum circuit compilation, which translates high-level quantum algorithms into low-level quantum instructions compatible with the target systems.

Quantum circuit compilation involves various tasks, including gate decomposition into the native operations of the quantum device, adaptation of operations to the hardware's topology, and optimization to reduce circuit depth, among others. Of particular importance is quantum circuit mapping, which aligns the quantum circuit with the hardware architecture, ensuring that any two-qubit operation is applied to physically connected qubits on the device. This alignment is often achieved through the insertion of SWAP operations, allowing to interchange the position of two logical qubits on the architecture. Chaining such operations permits arbitrary routing among remote qubits. To guarantee the reliable execution of the resulting circuit, it is imperative to minimize the overhead of SWAP operations. Previous research on the circuit mapping problem has predominantly focused on gate-optimal solutions, with the aim of minimizing the number of inserted SWAP gates~\cite{li2019tackling, siraichi2019qubit, siraichi2018qubit, wille2019mapping, zulehner2017exact, zulehner2018efficient}. Others have concentrated on time-optimal circuit mapping, which seeks to minimize the depth of the entire transformed circuit~\cite{tan2020optimal, lao2021timing, zhang2021time}.

The introduction of additional SWAP gates in quantum circuits aligns them with a target topology but simultaneously introduces more errors during their implementation. To address this, the authors in~\cite{nation2023suppressing} introduced quantum circuit remapping as a post-compilation step, aiming to reconfigure the initial mapping obtained from the standard techniques. This process was particularly realized by the MAPOMATIC algorithm that identifies topologically equivalent mappings but with enhanced fidelity on the target quantum system. It has also been integrated with existing circuit mapping techniques and has proven to be efficient in mitigating quantum circuit errors arising from system variability. However, in the remapping process, it requires to solve a subgraph matching problem, which is generally known to be NP-complete~\cite{cook71}. In the MAPOMATIC algorithm, this challenge was addressed by employing standard subgraph matching routines, such as VF2 and VF2++, to search for isomorphic graphs~\cite{cordella2004sub, juttner2018vf2++}. However, the time complexity of these procedures scales at least polynomially with the total number of qubits on the target quantum system. This polynomial scaling may present scalability challenges, especially in the context of the substantial advancements in quantum hardware that have been observed~\cite{arute2019quantum, wu2021strong}. The latest quantum computers already contain hundreds of qubits~\cite{collins2022ibm}, and this number is expected to increase by a thousandfold in the coming years~\cite{roadmap2022ibm}, eventually reaching millions of qubits for fully-fledged quantum computers. Given these visions and developments, current remapping algorithms may face challenges, even when compiling small circuits onto near-term quantum processors.

In this work, we address the scalability challenge of quantum circuit remapping by leveraging the inherent symmetries present in scalable quantum architectures. The main idea of this work is illustrated in Figure~\ref{fig:circuit_mapping_idea}. In the precompilation phase, we determine an initial circuit mapping that aligns the interaction graph of the compiled circuit with the coupling graph of the target quantum processor. Subsequently, during the remapping process, we identify all isomorphic subgraphs of the initial mapping and score the induced mappings using hardware calibration data. In particular, we introduce an algorithm that significantly reduces the search space for all topologically equivalent circuit mappings by leveraging hardware symmetries and efficiently identifies the optimal circuit mapping through a vectorized scoring method. Notably, this symmetry-based circuit remapping algorithm exhibits linear scaling with the number of qubits in the target hardware and is proved to be optimal in terms of its time complexity. Moreover, we benchmark the runtime of our algorithm with three typical quantum processors of grid, octagonal, heavy-hex structures, respectively, demonstrating the superior efficiency of our algorithm over the existing method in practical scenarios. To the best of our knowledge, this study marks the first instance of utilizing the topological symmetries of quantum processors to address quantum circuit compilation challenges, potentially offering insights into other compilation processes. Moreover, our technical contribution is a purely mathematical result and holds promise for application in various other problem domains.

\vspace{-0.2cm}
The remainder of this paper is structured as follows: Section~\ref{section: preliminaries} introduces the preliminaries to be used throughout this work. Section~\ref{section: symmetry-based quantum circuit mapping} presents the symmetry-based subgraph matching algorithm and its underlying theorem, justifying the search for all topologically equivalent circuit mappings within a reduced search space. It also outlines the vectorized scoring method that expedites the mapping process in practice. Section~\ref{section: performance benchmark} offers a comparative analysis of our algorithm against existing methods. Finally, Section~\ref{section: concluding remarks} concludes the paper and proposes potential directions for future research.

\section{Preliminaries}
\label{section: preliminaries}
\subsection{Notation}

\vspace{-0.4cm}
A {graph} is defined as an ordered pair of the vertex set and the edge set, denoted as $G = (V, E)$, where $V$ (or $V(G)$) denotes the set of vertices and $E$ (or $E(G)$) represents the set of edges. Each edge is an unordered pair of vertices $e = \{u, v\}$ for some $u, v \in V$. The order of a graph is its number of vertices $|V|$. A graph $G'$ is a subgraph of $G$, denoted as $G' \subseteq G$, if $V(G') \subseteq V(G)$ and $E(G') \subseteq E(G)$. An induced subgraph of a graph is another graph, formed from a subset of the vertices of the graph and all of the edges (from the original graph) connecting pairs of vertices in that subset. The degree of a vertex of a graph is the number of edges that are connected to the vertex. A graph has a bounded degree if the maximum degree of any vertex in the graph is limited by a constant value that remains independent of the graph's order. A {walk} is a sequence of edges $(e_{1}, e_{2}, \cdots, e_{n})$ for which there is a sequence of vertices $(v_{1}, v_{2}, \cdots, v_{n+1})$ such that $e_{k} = \{v_{k}, v_{k+1}\}$ for $k \in \{1, 2, \cdots, n\}$. A {path} is a walk in which all edges and all vertices are distinct, and the length of the path is defined as the number of edges it contains. Two vertices $u$ and $v$ are {connected} if there exists a path from vertex $u$ to vertex $v$ in graph $G$. Two vertices are adjacent if they are connected by a path of length $1$, i.e. by a single edge. A graph is connected if every pair of vertices in the graph is connected. The $k$-th order neighborhood of a vertex $v$ in a graph $G$, denoted as $N_{G}^{k}(v)$, refers to the set of vertices that can be reached from $v$ within $k$ hops. In other words, $u \in N_{G}^{k}(v)$ if and only if $u = v$ or there exists a path connecting $u$ and $v$ with a length no greater than $k$. The $k$-th order neighborhood of a subset $S$, denoted as $N_G^k(S)$, is defined as the union of the $k$-th order neighborhoods of all vertices within the subset $S$.

Given two graphs $G_{1} = (V_{1}, E_{1})$ and $G_{2} = (V_{2}, E_{2})$. These graphs are isomorphic, denoted as $G_{1} \cong G_{2}$, if there exists a bijection between the vertex sets $f: V_{1} \rightarrow V_{2}$ such that $\{u, v\} \in E_{1}$ if and only if $\{f(u), f(v)\} \in E_{2}$. This bijection $f$ is referred to as an {isomorphism} from $G_{1}$ to $G_{2}$. In particular, an {automorphism} is an isomorphism from a graph to itself. Note that $f$ is a mapping from vertices to vertices, but we can naturally extend it to a mapping that operates on both vertices and edges: $\tilde{f}: V_{1} \cup E_{1} \rightarrow V_{2} \cup E_{2}$, defined by $\tilde{f}(v) = f(v)$ and $\tilde{f}(\{u, v\}) = \{f(u), f(v)\}$. Therefore, $\tilde{f}$ can be taken as a mapping from graph to graph and is referred to as a natural extension of $f$.  
In this work, we will address the subgraph matching problem, which involves finding all subgraphs in a target graph that are isomorphic to a given pattern graph. This problem is NP-complete in general~\cite{cook71,qiao2017subgraph}. 

\subsection{Symmetry of graph}

In general, the symmetry of an object is associated with a transformation that preserves certain structural aspects of it. In graph theory, the automorphisms of a graph play the role of symmetry transformations. We denote the set of all automorphisms of a graph $G$ as $\mathrm{Sym}(G)$. This set forms a group under the composition of operators, which we refer to as the symmetry group of $G$.  Let $H$ be a subgroup of $\mathrm{Sym}(G)$. For a fixed vertex $x$ in $G$, the orbit of $x$ under $H$ is defined as $\mathrm{orb}_{H}(x) = \{y \in V(G): y = gx\ \mathrm{for\ some}\ g \in H\}$. In particular, we refer to the orbit of $x$ with respect to a symmetry group $H = \{f^{i}: i \in \mathbb{Z}\}$, as the orbit of $x$ with respect to automorphism $f$, denoted as $\mathrm{orb}_{f}(x)$. It is easy to show that all orbits within $G$ form a partition of the vertices, and two vertices $u$ and $v$ in $V(G)$ belong to the same orbit if and only if there exists an automorphism $f$ of graph $G$ such that $f(u) = v$. The orbit of a set of vertices $V_{0} \subseteq V(G)$, denoted as $\mathrm{orb}_{f}(V_{0})$, is defined by the union of the orbits $\mathrm{orb}_{f}(x)$ for all $x \in V_{0}$. A generating set of a graph $G = (V, E)$ with respect to the automorphism $f$ is a subset of vertices $V_{0} \subseteq V$ such that the orbit of $V_{0}$ with respect to $f$ covers all the vertices in $G$, i.e., $\mathrm{orb}_{f}(V_{0}) = V$. Similarly, the generating set of $G$ under a group of automorphisms $F$ is a subset of vertices $V_{0}' \subseteq V$ such that $\mathrm{orb}_{F}(V_{0}') = V$. 

\subsection{Graph distance}

The {distance} between two vertices $v$ and $u$ within a graph, denoted as $d(v, u)$, corresponds to the shortest path connecting them, measured in terms of the number of edges. In cases where no such path exists, the distance is conventionally regarded as infinite. The {eccentricity} $\epsilon(v)$ for a given vertex $v$ within a graph $G = (V, E)$ is defined as the greatest distance between vertex $v$ and any other vertex, that is, $\epsilon(v) = \max_{u \in V}d(v, u)$. This value tells how distant a node is from the farthest node in the entire graph. The {radius} of a graph $G = (V, E)$, denoted as $r(G)$, represents the minimum eccentricity among all vertices, that is, $r(G) = \min_{v \in V}\epsilon(v)$. Then the minimizer is called a {central vertex} of the graph. 

\vspace{-0.4cm}
\subsection{Quantum circuit mapping}

The interaction graph of a quantum circuit serves as a representation of the qubits and the required interactions between them. Each node in the interaction graph corresponds to a logical qubit in the quantum circuit, and an edge connects two qubits if there is a two-qubit gate in the circuit acting on both of those qubits. However, existing quantum computers typically provide a native gate set that includes a family of single-qubit gates along with some two-qubit gates. The interactions between qubits in a quantum processor are also constrained by the connectivity of its architecture, which can be expressed through the coupling graph, denoted as $G = (V, E)$. In this graph, $V$ represents the set of physical qubits, and an edge $\{u, v\} \in E$ signifies that a two-qubit gate can be directly executed between physical qubits $u$ and $v$.

To execute a quantum circuit on a target quantum device, the set of logical qubits, denoted as $Q_l$, in the interaction graph of the circuit must first be mapped to the physical qubits, represented as $Q_p$, according to the coupling graph of the quantum processor. This mapping involves establishing an injective function $g: Q_l \to Q_p$, meaning that each logical qubit is uniquely assigned to a specific physical qubit. The purpose of circuit mapping is to ensure that the two-qubit gates in the circuit can be executed using the available physical qubits in accordance with the constraints defined by the coupling graph of the quantum device. In graph theory, this is the same as finding a subgraph of the coupling graph that is isomorphic to the interaction graph of the circuit.

Quantum circuit remapping is introduced as a subsequent step that follows circuit mapping and entails the mapping of $h: Q_p \to Q_p$. Therefore, the combined effect of the circuit mapping and remapping, denoted as $h \circ g$, is a mapping from logical qubits to physical qubits. This mapping aims to reduce errors during the implementation of the quantum circuit. The most crucial step in the remapping process is to identify all subgraphs in the coupling graph that are isomorphic to the interaction graph of the circuit, essentially solving a subgraph matching problem. 

% \vspace{-0.2cm}
\section{Symmetry-Based Quantum Circuit Mapping}
\label{section: symmetry-based quantum circuit mapping} 

Symmetries play a pivotal role in various applications within the field of quantum information science, including designing quantum machine learning algorithms~\cite{lyu2023symmetry,meyer2023exploiting}, evaluating quantum channel capacities~\cite{wang2018semidefinite,fang2021geometric,fang2019quantum}, analyzing quantum entropies~\cite{fang2021ultimate,Fang2020} and quantifying all kinds of quantum resources for quantum computation and quantum communication~\cite{hayashi2021finite,Fang2022,Fang2020nogo,regula2019one,fang2019non,fang2021sum,Diaz2018usingreusing,fang2018,regula2018,xie2017}. Nevertheless, the potential of harnessing hardware symmetries for quantum circuit compilation remains largely unexplored.

In this section, we introduce a subgraph matching algorithm that leverages the topological symmetry of quantum processors. We also present a circuit mapping scoring algorithm that utilizes vector computation. These routines are then integrated into a compilation pipeline, forming a complete quantum circuit mapping algorithm. We then provide a complexity analysis of our algorithm, highlighting its optimality with respect to time complexity.

\subsection{Symmetry-based subgraph matching}

The idea behind our subgraph matching algorithm, as depicted in Figure~\ref{fig:circuit_mapping_idea}, is to narrow down the search space for isomorphic graphs by focusing on a neighborhood of the generating set within the target graph. This approach is supported by the following Theorem~\ref{theorem: core}, which demonstrates that all isomorphic graphs can be obtained by applying the corresponding symmetry transformations to graphs within the reduced space.

\begin{theorem}
Let $T$ be a graph and $S$ be a generating set of $T$ with respect to the automorphism $f$. Let $G$ be a subgraph of $T$ with radius $r$. Then for any subgraph $G''$ of $T$ that is isomorphic to $G$, there exists an integer $n$ and subgraph $G'\subseteq T$ such that $G'' = \tilde f^n(G')$ with $G' \cong G$, $V(G') \subseteq N^{r}_{T}(S)$ and $\tilde f$ being the natural extension of $f$. \label{theorem: core}
\end{theorem}

The proof of this result requires the following lemmas. 

\begin{lemma}\label{theorem: connected_path}
Let $T$ be a graph with an automorphism $f$. Let $u, v$ be two vertices in $T$. Then there exists a path connecting $u$ and $v$ of length $l$ if and only if $f^{n}(u)$ and $f^{n}(v)$ are connected by a path of length $l$ for any integer $n$. 
\end{lemma}

\begin{proof}
Suppose that $(w_0, w_{1}, \cdots, w_l)$ is a path connecting vertices $w_0 = u$ and $w_l = v$. Then $\{w_k, w_{k+1}\}$ is an edge in graph $T$ for any $k \in \{0, \cdots, l-1\}$. Since $f$ is an automorphism, we know that $\{f(w_k), f(w_{k+1})\}$ is also an edge in $T$. This implies that $(f(w_0), f(w_{1}), \cdots, f(w_l))$ is a path in $T$, which connects $f(u)$ and $f(v)$ and has a length of $l$. By repeating this process, we can demonstrate that for all positive integers $n$, $f^n(u)$ and $f^n(v)$ are connected by a path of length $l$. Furthermore, since $f^{-1}$ is also an automorphism of $T$, we can extend the range of $n$ to all integers and conclude the proof.
\end{proof}

The next lemma demonstrates that the operations that take the neighborhood of a set and map a set under automorphisms can be exchanged.
  
\begin{lemma} 
Let $T$ be a graph with an automorphism $f$, and let $S \subseteq V(T)$ be a subset of vertices. Then for any integer $n$ and positive integer $m$, $N^{m}_{T}(f^{n}(S)) = f^{n}(N^{m}_{T}(S))$. \label{theorem: commutability}
\end{lemma}

\begin{proof}
For any $x \in N^{m}_{T}(f^{n}(S))$, there exists $z$ in $S$ and $y = f^{n}(z)$ such that $x$ and $y$ are connected by a path of length $l$ in $T$ with $l \le m$. By Lemma~\ref{theorem: connected_path}, we know that there exists a path of length $l$ connecting $f^{-n}(x)$ and $z = f^{-n}(y)$. This implies $f^{-n}(x) \in N^{m}_{T}(S)$. Hence, there exists $s$ in $N^{m}_{T}(S)$ such that $f^{-n}(x) = s$, or equivalently $x = f^n(s)$. This gives $x \in f^n(N^{m}_{T}(S))$ and concludes that $N^{m}_{T}(f^{n}(S)) \subseteq f^n(N^{m}_{T}(S))$. Now we prove the other direction. Consider any $x \in f^n(N^{m}_{T}(S))$. There exists $s$ in $N^{m}_{T}(S)$ such that $x = f^n(s)$. Since $s$ is in $N^{m}_{T}(S)$, we know that there exists $y$ in $S$ and a path of length $l$ connecting $s$ and $y$ with $l\leq m$. By Lemma~\ref{theorem: connected_path}, we know that there exists a path of length $l$ connecting $x = f^n(s)$ and $f^n(y) \in f^n(S)$. This implies $x \in N^{m}_{T}(f^{n}(S))$ and concludes that $f^n(N^{m}_{T}(S)) \subseteq N^{m}_{T}(f^{n}(S))$.
\end{proof}

We are now ready to prove the main theorem.

\vspace{0.2cm}

\noindent \textbf{Proof of Theorem~\ref{theorem: core}} Since $G''$ is isomorphic to $G$, they have the same radius $r$. Let $v$ be a central vertex of $V(G'')$. According to the definition of generating set, there exists an integer $n$ such that $v \in f^{n}(S)$. By the definition of central vertices, all vertices of $G''$ must be reachable from $v$ within $r$ hops. That is, $G''$ is contained by the $r$-th order neighborhood of $v$, namely $V(G'') \subseteq N_{T}^{r}(v)$. Since $v \in f^{n}(S)$, we then get $V(G'') \subseteq N_{T}^{r}(f^{n}(S))$. By Lemma~\ref{theorem: commutability}, we have $N_{T}^{r}(f^{n}(S)) = f^{n}(N_{T}^{r}(S))$. This gives $V(G'') \subseteq f^{n}(N_{T}^{r}(S))$. Now let us choose $G' = \tilde{f}^{-n}(G'')$, we have $V(G') = V(\tilde{f}^{-n}(G'')) = f^{-n}(V(G'')) \subseteq f^{-n}(f^{n}(N_{T}^{r}(S))) = N_{T}^{r}(S)$, where the second equality follows by the definition of the natural extension of an automorphism. Since $\tilde{f}$  maps graphs to isomorphic graphs, we have $G' = \tilde{f}^{-n}(G'') \cong G''$. Note that $G'' \cong G$, we get $G' \cong G$. This concludes the proof.
\hfill $\blacksquare$ \vspace{0.2cm}

The automorphism in Theorem~\ref{theorem: core} can be replaced with a group of automorphisms. This result is formally presented in Corollary~\ref{corollary: core} and applies to various practical scenarios where a graph has multiple independent symmetry transformations. For example, a two-dimensional grid possesses two independent symmetry transformations: translation in both the horizontal and vertical directions. 

\begin{corollary}
Let $T$ be a graph and $S$ be a generating set of $T$ with respect to a group of automorphisms $F$. Let $G$ be a subgraph of $T$ with radius $r$. Then for any subgraph $G''$ of $T$ that is isomorphic to $G$, there exists an automorphisms $f \in F$ and subgraph $G'\subseteq T$ such that $G'' = \tilde f(G')$ with $G' \cong G$, $V(G') \subseteq N^{r}_{T}(S)$ and $\tilde f$ being the natural extension of $f$. \label{corollary: core}
\end{corollary}

Theorem~\ref{theorem: core} and Corollary~\ref{corollary: core} indicate that to locate all the subgraphs in $T$ that are isomorphic to $G$, one only needs to search over the subgraph induced by the vertex set $N_{T}^{r}(S)$.
In light of this reduction, we propose the symmetry-based subgraph matching (SBSM) algorithm (Algorithm~\ref{algorithm: subgraph matching}), enabling efficient searching for all circuit mappings given a coupling graph.

% \vspace{0.2cm}
\begin{algorithm}[!ht]
\caption{Symmetry-Based Subgraph Matching (SBSM)}
\label{algorithm: subgraph matching}

\SetKwInOut{Input}{Input}
\SetKwInOut{Output}{Output}
\LinesNumbered

\Input{
A pattern graph $G$, a target graph $T$, a group of automorphisms $F$ of the target graph and the associated generating set $S$.
}

\BlankLine

\Output{
All subgraph isomorphisms of $G$ in $T$.
}

\BlankLine
\BlankLine

Let $r$ be the radius of the pattern graph $G$;

Let $N_{T}^{r}(S)$ be the $r$-th order neighborhood of $S$;

Let $R$ be the induced subgraph of $N_{T}^{r}(S)$ in $T$;

Obtain the set of all isomorphic graphs of $G$ within $R$ and denote this set as $H_{0}$; 

Let $H$ be an empty list;

\For{$G' \in H_{0}$}{
    Let $M = \{\tilde{f}(G'): f \in F, \tilde{f}(G') \subseteq T\}$; 
    
    Append $M$ to $H$;
}

Return $H$;

\end{algorithm}
% \vspace{0.2cm}

The SBSM algorithm consists of three major steps. First, the algorithm identifies a reduced search space for initial subgraph matching. To this end, the algorithm calculates the radius of $G$, denoted as $r$, and computes the $r$-th order neighborhood of $G$, serving as the reduced search space. Second, it employs a standard subgraph matching algorithm, such as VF2~\cite{cordella2004sub}, to search for isomorphic graphs within the reduced search space. Lastly, the searched patterns will be transformed using the symmetry transformations to get all the isomorphisms. 

Note that the coupling graph of a quantum processor may not inherently possess symmetries, particularly along its boundary. Nevertheless, we can effectively apply the SBSM algorithm by embedding the coupling graph within a larger graph that does possess symmetries. It is also worth mentioning that the SBSM algorithm can accommodate various types of symmetries, which may include but are not limited to translational, mirror, and rotational symmetries. For the numerical experiments discussed in the following sections, we will primarily focus on translational symmetries for simplicity. 

\subsection{Circuit mapping scoring through vectorization}

Quantum circuit remapping involves the assessment of all topologically equivalent circuit mappings. The existing approach in Qiskit accomplishes this by using a for-loop to estimate the overall circuit fidelity for each mapping one by one~\cite{Qiskit}. However, the for-loop implementation is time-consuming, particularly when dealing with a large number of mappings. To mitigate the computational cost of scoring these mappings, we ultilize the technique of vectorization, which transforms the evaluation process into vector computations. As we will demonstrate later, the adoption of vectorization substantially speeds up the circuit evaluation process, which would otherwise be dominant in the runtime of our remapping algorithm. 

Our circuit mapping scoring algorithm is presented in Algorithm~\ref{algorithm: scoring circuit mapping}, where we emphasize the vectors in bold symbols. In this algorithm, we denote the circuit mappings and the error map by $\boldsymbol{L}$ and $E$, respectively. The initial phase involves the transformation of the precompiled circuit into a sequence of individual logical gates, represented as $\boldsymbol{G}$. For each logical gate $g$ within $\boldsymbol{G}$, we access the error rate vector of the corresponding physical gates from the hardware calibration data, $\boldsymbol{E}' = E(\boldsymbol{L}(g))$. Note that $\boldsymbol{E}'$ encapsulates the error rates associated with the specific gate $G'$ across all circuit mappings. Subsequently, by harnessing the element-wise vector multiplication supported by the NumPy library~\footnote{https://numpy.org/}, we perform vector computations to evaluate the circuit scoring vector $\boldsymbol{S}$, recording the estimated fidelity values of all circuit mappings. That is, we perform element-wise multiplication, denoted by $\odot$, iteratively on $\boldsymbol{1} - \boldsymbol{E}'$ for all $g$, where $\boldsymbol{1}$ is the vector of ones. The resulting $\boldsymbol{S}$ is used to guide the selection of the optimal mapping that maximizes the overall circuit fidelity.

\vspace{0.2cm}
\begin{algorithm}[!ht]
    \caption{Circuit Mapping Scoring through Vectorization}\label{algorithm: scoring circuit mapping}
    
    \SetKwInOut{Input}{Input}
    \SetKwInOut{Output}{Output}
    \LinesNumbered
    
    \Input{
    An array of circuit mappings $\boldsymbol{L}$, a logical quantum circuit $C$, and an error map $E$.
    }
    
    \BlankLine
    
    \Output{
    The estimated fidelity of the quantum circuit under the circuit mappings.
    }
    
    \BlankLine
    \BlankLine
    
    Convert $C$ to a list of gates $\boldsymbol{G}$;
    
    Initiate $\boldsymbol{S}$ to be a vector of ones $\boldsymbol{1}$;
    
    \For{$g$ in $\boldsymbol{G}$}{
        Compute $\boldsymbol{L}' = \boldsymbol{L}(g)$;
    
        Compute $\boldsymbol{E}' = E(\boldsymbol{L}')$; 
    
        Update $\boldsymbol{S}$ as  $\boldsymbol{S} \odot (\boldsymbol{1} - \boldsymbol{E}')$;
    }
    
    Return $\boldsymbol{S}$;
\end{algorithm}
\vspace{0.2cm}

\subsection{Symmetry-based quantum circuit mapping}

The SBSM algorithm and the circuit mapping scoring algorithm can work together to facilitate efficient quantum circuit remapping. Now we integrate these steps into a complete quantum circuit mapping algorithm, as presented in Algorithm~\ref{algorithm: circuit mapping}. This symmetry-based circuit mapping (SBCM) algorithm comprises several  main steps. First, the algorithm ultilizes existing methods, such as those supported by Qiskit, to compile the quantum circuit, obtaining the interaction graph $G$ and an initial circuit mapping $L_\text{pre}$. Note that this precompilation should be chosen as independent of the order of $T$. Second, the SBSM algorithm is used to search for isomorphic subgraphs of $G$ within the hardware coupling graph $T$. These discovered subgraphs are recorded in the vector $\boldsymbol{L}$. Note that each isomorphic subgraph is associated with a circuit mapping from the precompiled circuit onto the given quantum hardware. Third, the algorithm employs vectorized computation, as stated in Algorithm~\ref{algorithm: scoring circuit mapping}, to score all circuit mappings. Finally, the circuit mappings are sorted based on their scores (the estimated circuit fidelity), and the circuit mapping that maximizes the scores is selected and returned.

\vspace{0.2cm}
\begin{algorithm}[ht]
\caption{Symmetry-Based Circuit Mapping (SBCM)}\label{algorithm: circuit mapping}

\SetKwInOut{Input}{Input}
\SetKwInOut{Output}{Output}
\LinesNumbered

\Input{
A logical quantum circuit $C$, a coupling graph $T$ of the target quantum processor, a group of automorphisms $F$ of the coupling graph and the associated generating set $S$, and an error map $E$ of the quantum processor.
}

\BlankLine

\Output{
A compiled quantum circuit and its circuit mapping.
}

\BlankLine
\BlankLine

Use $C, T, E$ as inputs to precompile the circuit for selected physical qubits and obtain a precompiled circuit $C'$, its interaction graph $G$ and an initial mapping $L_{\text{pre}}$;

Use $T, F, S, G$ as the inputs of Algorithm~\ref{algorithm: subgraph matching} to get all isomorphic subgraphs of $G$ in $T$ and record them in $\boldsymbol{L}$; 

Compute the element-wise composition $\tilde{\boldsymbol{L}} = \boldsymbol{L} \circ L_{\mathrm{pre}}$;

Use $\tilde{\boldsymbol{L}}, C', E$ as the inputs of Algorithm~\ref{algorithm: scoring circuit mapping} to score the circuit mappings and record the mapping-and-score pair in $(\tilde{\boldsymbol{L}}, \boldsymbol{S})$;

Return the compiled circuit $C'$ and the circuit mapping with the highest score;
\end{algorithm}
\vspace{0.2cm}

\begin{figure*}[!ht]
    \centering
    \includegraphics[width=0.8\textwidth]{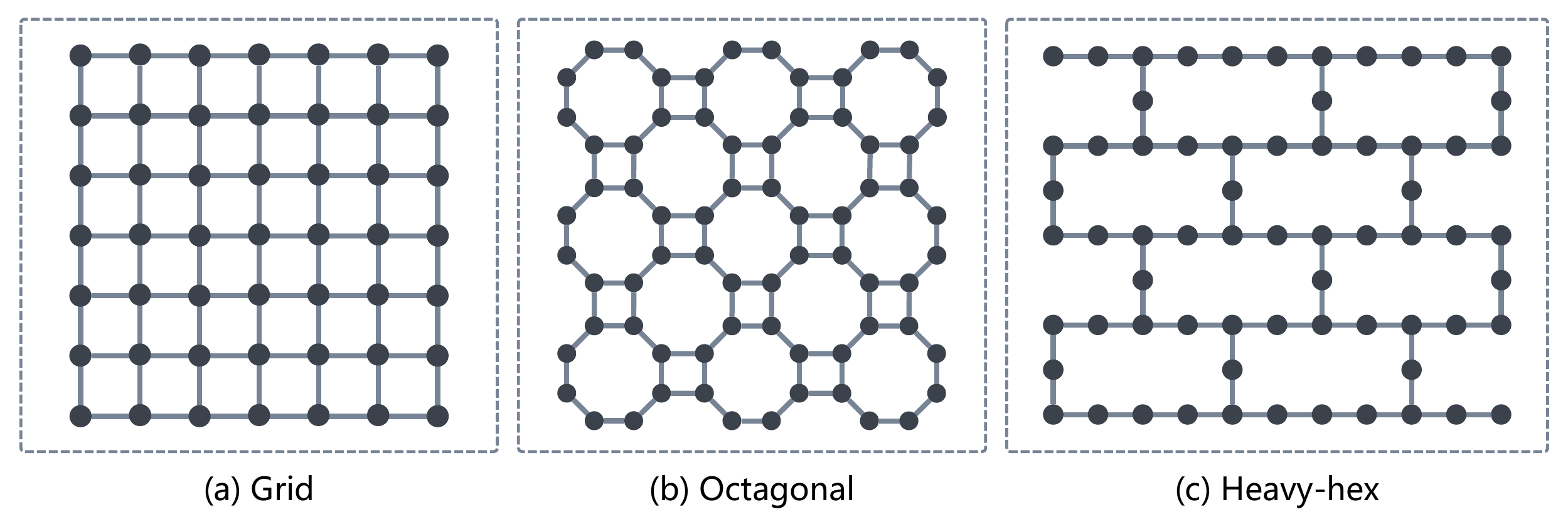}
     \caption{\footnotesize Three typical quantum chip structures. Each node represents a physical qubit, and each edge corresponds to a coupler between two qubits. (a) $7 \times 7$ grid coupling graph with $49$ qubits. (b) $3 \times 3$ octagonal coupling graph with $72$ qubits. (c) $4 \times 2$ heavy-hex coupling graph with $67$ qubits.}
    \label{fig:chip_topology}
\end{figure*}

\subsection{Computational complexity}

We conduct a computational complexity analysis for our symmetry-based algorithms, with a particular focus on scalable quantum processors that feature sparse coupling graphs with bounded degrees. That is, the maximum degree of any vertex in the graph is constrained by a constant value that remains independent of the graph's order. This ensures that the size of the reduced search space remains constant. We are aware that this constraint may not hold for trapped-ion quantum systems whose coupling graph is complete. Nevertheless, such systems inherently exhibit scalability issues, and we believe that the bounded degree assumption is reasonable for all scalable quantum processors, especially in the context of superconducting quantum systems. Our complexity analysis focuses on the scaling of the coupling graph, while we consider the size of the circuit to compile as fixed. 

The following result shows that the SBSM algorithm exhibits a time complexity of $O(n)$, which stands in stark contrast to VF2, a widely used subgraph matching algorithm with a time complexity of $O(n^{2})$ in the best-case scenario and a much less efficient time complexity of $O(n!n)$ in the worst-case scenario~\cite{cordella2004sub}.

\begin{theorem}
Let $T$ be a target graph with order $n$ that has a bounded degree and is characterized by a group of automorphisms $F$ along with a constant-size generating set $S$. Then the SBSM algorithm has a time complexity of $O(n)$, which is optimal for the subgraph matching problem.
\label{thm: SBSM theorem}
\end{theorem}

\begin{proof}
The subgraph matching routine consists of three steps. First, we determine the reduced search space $N_T^r(S)$, which is the $r$-th order neighborhood of a generating set $S$ and $r$ is the radius of the pattern graph. Calculating the graph radius is efficiently performed using a breadth-first search algorithm, with a complexity that is independent of the order of the target graph~\cite{roditty2013fast}. Similarly, as the cardinality of the generating set is constant, identifying its $r$-th order neighborhood requires at most $O(n)$ steps. 
Second, subgraph matching is conducted within this reduced search space, and we consider utilizing the VF2 algorithm, known for its efficiency in subgraph matching and widely used in applications such as graph isomorphism tests within NetworkX~\cite{networkx2023isomorphism}. Since the target graph has a bounded degree, the size of the reduced search space is independent of $n$. In this case, the time complexity of subgraph matching using VF2 in the reduced space remains constant~\cite{cordella2004sub}.
Finally, symmetry transformations are applied to obtain all isomorphic subgraphs. By the subsequent Lemma~\ref{lem: linear iso}, there are at most $O(n)$ subgraphs isomorphic to a connected pattern graph in $T$. So finding all isomorphic subgraphs requires only a number of steps linear in $n$. As a result, the SBSM algorithm has an overall time complexity of $O(n)$. Note that the subgraph matching problem entails finding all isomorphic subgraphs in the target graph, necessitating a minimum of visiting each vertex at least once, resulting in a lower bound of $\Omega(n)$ steps in total. Our SBSM algorithm meets this trivial lower bound, thus concluding the proof of its optimality.  
\end{proof}

\begin{lemma}
Let $T$ be a target graph with order $n$ that has a bounded degree of $d$. Let $G$ be a connected pattern graph with order $m$. Then there are at most $O(n)$  subgraphs in the target graph $T$ that are isomorphic to $G$.
\label{lem: linear iso}
\end{lemma}

\begin{proof}
A subgraph isomorphism $f: G \rightarrow T$ embeds the vertex set $V(G)$ into $V(T)$. Let $v_1$ be a vertex in $G$. There are $n$ possible values for $f(v_{1})$, corresponding to all vertices in $T$. Once $f(v_{1})$ is fixed, we choose a vertex $v_2 \in V(G)\backslash \{v_1\}$ that is adjacent to $v_{1}$. The existence of $v_2$ is ensured by the connectivity of $G$. Since $f$ is a subgraph isomorphism from $G$ to $T$, $f(v_{2})$ must be a vertex in $T$ adjacent to $f(v_{1})$. As $T$ has a bounded degree of $d$, each vertex in $T$ has at most $d$ adjacent vertices. Therefore, there are at most $d$ possible values for $f(v_{2})$. Continuing this process, we consider $f(v_{3})$, where $v_{3} \in V(G)\backslash \{v_1,v_2\}$ and is adjacent to at least one node in $\{v_{1}, v_{2}\}$. Without loss of generality, we assume that $v_{2}$ is adjacent to $v_{3}$ for simplicity. Then $f(v_{3})$ must be a vertex in $T$ adjacent to $f(v_{2})$ because $f$ is a subgraph isomorphism. As $T$ has a bounded degree of $d$, there are at most $d$ possible values for $f(v_{3})$. By interating this procedure, we can traverse all the vertices in $G$ and find their images under $f$, concluding that there are at most $nd^{m-1}$ ways to match vertices that ensure $f$ being a subgraph isomorphism. In other words, the total number of subgraph isomorphisms from $G$ to $T$ is at most $nd^{m-1}$ with the dominant factor $O(n)$.
\end{proof}

Theorem~\ref{thm: SBSM theorem} demonstrates that our SBSM algorithm exhibits an optimal time complexity for solving the subgraph matching problem with symmetries, making it an ideal choice for the remapping step. This optimality extends to quantum circuit mapping consisting of the standard mapping and the incorporation of an additional remapping process, as exemplified in Algorithm~\ref{algorithm: circuit mapping}. This is because the initial mapping is independent of the size of the target quantum processor, and the remapping process dominates the overall circuit mapping process in the runtime. This optimality is formally presented in Corollary~\ref{theorem: optimality_SBCM}.

\begin{figure*}[!htbp]
    \centering
    \includegraphics[]{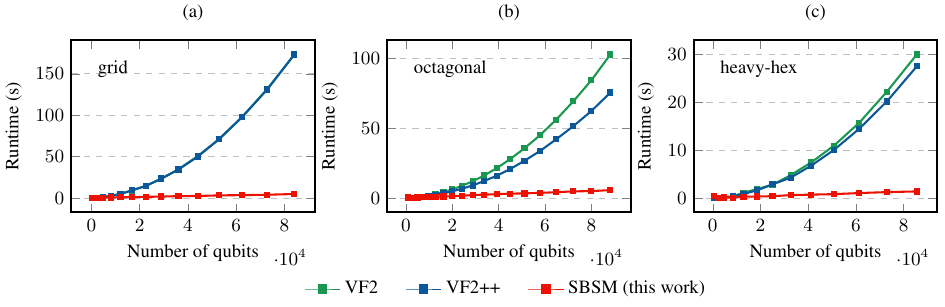}
    \caption{\footnotesize The runtime of three subgraph matching algorithms: VF2 (green)~\cite{cordella2004sub}, VF2++ (blue)~\cite{juttner2018vf2++}, and the SBSM algorithm (red). These experiments solved the subgraph matching problem on three quantum hardware architectures. Note that the results for VF2 and VF2++ coincide in the first plot.}
    \label{fig:subgraph_matching}
\end{figure*}

\begin{figure*}[!htbp]
    \centering
    \includegraphics[]{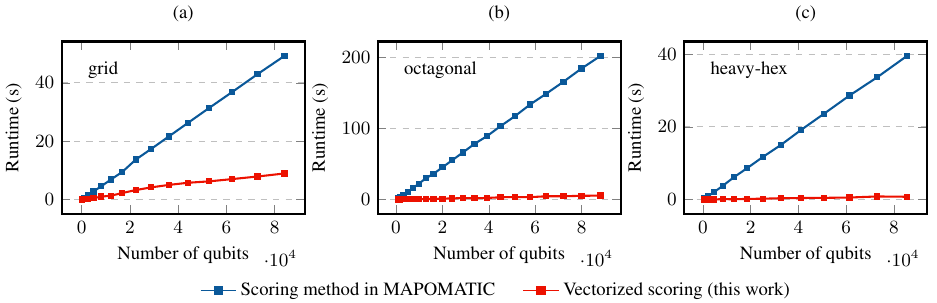}
    \caption{\footnotesize The runtime for scoring circuit mappings using the conventional for-loop method in MAPOMATIC (blue)~\cite{nation2023suppressing} and our vectorized approach (red). These experiments were conducted on three distinct quantum hardware architectures.}
    \label{fig:runtime_scoring}
\end{figure*}

\begin{figure*}[!htbp]
    \centering
    \includegraphics[]{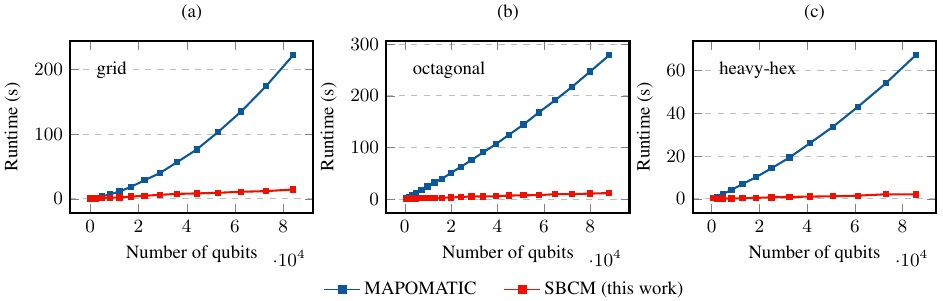}
    \caption{\footnotesize The runtime of two circuit remapping algorithms: MAPOMATIC (blue)~\cite{nation2023suppressing} and SBCM (red). These experiments were conducted on three distinct quantum hardware architectures.}
    \label{fig:circuit_mapping}
\end{figure*}

\begin{corollary}
Consider a quantum processor with a coupling graph of order $n$ that has a bounded degree and is characterized by a group of automorphisms $F$ along with a constant-size generating set $S$. Let $C$ be a quantum circuit to compile. Then any quantum circuit mapping scheme that comprises circuit mapping and remapping steps requires at least $\Omega(n)$ time. Moreover, the SBCM algorithm achieves this optimal time complexity.
\label{theorem: optimality_SBCM}
\end{corollary}

\begin{proof}
Since the remapping process requires at least $\Omega(n)$ steps, the overall mapping has a time complexity no less than $\Omega(n)$. In terms of our SBCM algorithm, the precompilation is independent of the order of the coupling graph and it only requires $O(n)$ steps to find all isomorphic graphs. By Lemma~\ref{lem: linear iso}, there are at most $O(n)$ subgraphs isomorphic to a connected pattern graph in $T$. This implies that the scoring process and the search for the mapping with the highest score each consume a maximum of $O(n)$ steps. Hence, the SBCM algorithm has the time complexity of $O(n)$ and it achieves the lower bound.
\end{proof}

\vspace{-0.6cm}
\section{Benchmarking Results}
\label{section: performance benchmark}

In this section, we conduct numerical experiments to assess the performance of our algorithms and compare them to the existing ones such as VF2, VF2++ and MAPOMATIC. In particular, we consider quantum processors with grid, octagonal and heavy-hex architectures, as depicited in Figure~\ref{fig:chip_topology}, which are the state-of-the-art architectures ultilized by Google~\cite{arute2019quantum}, Rigetti~\cite{architecture2022rigetti} and IBM~\cite{architecture2021ibm}, respectively. We specifically consider layouts with equal rows and columns. The test circuit we use is the $5$-qubit Deutsch-Jozsa algorithm~\cite{deutsch1992rapid}, which is the pioneering example of a quantum algorithm that
outperforms classical counterparts and exemplifies the advantages of using quantum computing
for specific problems. All of the numerical experiments were performed on a MacBook Pro (2020) with an Intel i5 processor and 8 GB memory. The source codes are available on arXiv.

\vspace{-0.4cm}
\subsection{Benchmarking the subgraph matching algorithms}

We compared the runtime of the VF2, VF2++, and our SBSM algorithm for solving the subgraph matching problem. The results are presented in Figure~\ref{fig:subgraph_matching}, where the horizontal axis represents the size of the quantum processors, and the vertical axis indicates the runtime for identifying all isomorphic graphs within the target graph. Notably, as the size of the hardware increases, the runtime of VF2 and VF2++ experiences significant growth due to the expansion of the search space for subgraph matching. In contrast, the SBSM algorithm's runtime exhibits a gradual increase with the number of qubits in the hardware.
Specifically, for an $88,200$-qubit octagonal-shaped quantum processor, the SBSM algorithm runs approximately $18$ times faster than VF2 and $13$ times faster than VF2++. These benchmarking results underscore a remarkable performance advantage of the SBSM algorithm over VF2 and VF2++, making it a more suitable choice for the search of circuit mappings to large-scale quantum computers.

\subsection{Benchmarking the circuit mapping scoring methods}

We conducted benchmark tests comparing the runtime of both the conventional for-loop method employed by MAPOMATIC and our vectorized method for evaluating all mappings of a $5$-qubit Deutsch-Jozsa algorithm on the three quantum hardware architectures. 
To account for the impact of imperfect fabrication and noisy environments, we assumed that the fidelity of quantum gates follows a two-dimensional spatial normal distribution, with the highest gate fidelity assumed to be at the center of the layout. The results of the numerical experiments are presented in Figure~\ref{fig:runtime_scoring}, where the blue line represents the runtime of the scoring method used by MAPOMATIC, and the red line shows the runtime of our vectorized scoring. Evidently, for all three hardware architectures, the vectorized method's runtime is substantially shorter than that of the conventional for-loop scoring method. Specifically, in the case of an octagonal-shaped quantum processor with $88,200$ qubits, the vectorized method accomplishes this task approximately $36$ times faster than the conventional scoring function, resulting in an impressive $97\%$ reduction in runtime. This comparison illustrates the efficiency gains achieved through vectorization and the substantial time-saving potential in the circuit mapping process.

\subsection{Benchmarking the complete circuit mapping algorithms}

We conducted a benchmark for the complete quantum circuit mapping problem, comparing the total runtime of our SBCM algorithm with MAPOMATIC, as they map the $5$-qubit Deutsch-Jozsa algorithm onto the three quantum hardware architectures. The results are depicted in Figure~\ref{fig:circuit_mapping}, where the red line represents the runtime of the SBCM algorithm and the blue line represents that of MAPOMATIC. Notably, our SBCM algorithm outperforms MAPOMATIC for all three hardware architectures, with the runtime of the former being substantially smaller. For instance, when mapping the input circuit onto an $88,200$-qubit octagonal-shaped quantum processor, MAPOMATIC took approximately $279$ seconds to find the optimal circuit mapping, while our SBCM algorithm completed the same task in just $11$ seconds. This represents an impressive $96\%$ reduction in runtime. The comparison results clearly demonstrate the superiority of the SBCM algorithm over MAPOMATIC for mapping circuits onto quantum hardware with typical architectures, underscoring the significance of utilizing hardware symmetries in finding the optimal circuit mapping.

\section{Conclusions}
\label{section: concluding remarks}

In this work, we have introduced an efficient algorithm for identifying all circuit mappings within scalable quantum systems by leveraging its inherent symmetries. This algorithm is based on a subgraph matching approach that harnesses hardware symmetries to significantly reduce the size of the search space. We have provided theoretical proof of the optimality of the symmetry-based algorithms and conducted numerical experiments to confirm their advantages in practical scenarios. It is worth noting that the symmetry-based circuit mapping can be integrated with other existing compilation techniques~\cite{fang2023dynamic} and helps to bridge the gap between theoretical quantum algorithms and their physical implementation on quantum computers with operational constraints and limited resources.

As scalable quantum processors necessarily possess symmetries, we envision the potential to leverage hardware symmetries across various aspects of quantum circuit compilation and beyond. For instance, the intrinsic symmetries of quantum devices could be exploited to enhance the efficiency of quantum gate optimizations, dynamic circuit compilation, and distributed quantum computing across multiple quantum processors. We believe that ample opportunities for future research exist to delve deeper into these promising possibilities.

\vspace{-0.4cm}
\section*{Acknowledgements}
D. Y. acknowledges valuable discussions with Zhiping Liu. This work was done when D. Y. was a research intern at Baidu Research.

\bibliography{reference}

\end{document}

%% file: pretex.tex
\usepackage{times}

\usepackage[paperwidth=199.8mm,paperheight=297mm,centering,hmargin=15mm,vmargin=2cm]{geometry}
\usepackage[titletoc,title]{appendix}
\usepackage{hyperref}
\hypersetup{colorlinks=true,citecolor=blue,linkcolor=blue,filecolor=blue,urlcolor=blue,breaklinks=true}
\usepackage[dvipsnames]{xcolor}
\usepackage{color}
\usepackage[marginal]{footmisc}
\usepackage{url}
\usepackage{mathtools}
\usepackage{amsmath}
\usepackage{amssymb}
\usepackage[shortlabels]{enumitem}
\usepackage{graphicx,epic,eepic,epsfig,latexsym,verbatim}
\usepackage{dsfont}

\usepackage{framed}
\usepackage{caption}
\usepackage{float}
\usepackage{tikz}
\usepackage{tcolorbox}
\usepackage{relsize}

\definecolor{shadecolor}{rgb}{0.9,0.9,0.9}
\definecolor{LightGray}{RGB}{220,220,220}
\definecolor{LinkColor}{RGB}{167,20,49}
\definecolor{myred}{RGB}{236, 17, 0}
\definecolor{myblue}{RGB}{10, 88, 153}
\definecolor{myorange}{RGB}{236, 137, 0}
\definecolor{mygreen}{RGB}{26, 152, 81}

\usepackage{theorem}
\newtheorem{definition}{Definition}

\newtheorem{lemma}[definition]{Lemma}
\newtheorem{theorem}[definition]{Theorem}
\newtheorem{corollary}[definition]{Corollary}

\def\squareforqed{\hbox{\rlap{$\sqcap$}$\sqcup$}}
\def\qed{\ifmmode\squareforqed\else{\unskip\nobreak\hfil
\penalty50\hskip1em\null\nobreak\hfil\squareforqed
\parfillskip=0pt\finalhyphendemerits=0\endgraf}\fi}
\def\endenv{\ifmmode\;\else{\unskip\nobreak\hfil
\penalty50\hskip1em\null\nobreak\hfil\;
\parfillskip=0pt\finalhyphendemerits=0\endgraf}\fi}
\newenvironment{proof}{\noindent \textbf{{Proof~} }}{\hfill $\blacksquare$ \vspace{0.2cm}}

\newcounter{remark}

\newcounter{example}

\mathchardef\ordinarycolon\mathcode`\:
\mathcode`\:=\string"8000
\def\vcentcolon{\mathrel{\mathop\ordinarycolon}}
\begingroup \catcode`\:=\active
  \lowercase{\endgroup
  \let :\vcentcolon
  }

\usepackage{colortbl}
\usepackage{pifont}% http://ctan.org/pkg/pifont
\definecolor{Gray}{gray}{0.92}
\definecolor{Gray2}{gray}{0.75}
\definecolor{maroon}{cmyk}{0,0.87,0.68,0.32}
\usepackage{booktabs}
\usepackage{makecell}%To keep spacing of text in tables
\usepackage{diagbox}
\usepackage{multirow}

\newcommand{\nc}{\newcommand}
\nc{\rnc}{\renewcommand}
\nc{\bra}[1]{\langle#1|}
\nc{\ket}[1]{|#1\rangle}
\nc{\<}{\langle}
\rnc{\>}{\rangle}
\nc{\ketbra}[2]{|#1\rangle\!\langle#2|}
\nc{\braket}[2]{\langle#1|#2\rangle}
\nc{\braandket}[3]{\langle #1|#2|#3\rangle}
\nc{\proj}[1]{| #1\rangle\!\langle #1 |}
\nc{\avg}[1]{\langle#1\rangle}
\nc{\Rank}{\operatorname{Rank}}
\nc{\smfrac}[2]{\mbox{$\frac{#1}{#2}$}}
\nc{\tr}{\operatorname{Tr}}
\nc{\ox}{\otimes}

\nc{\cA}{{\cal A}}
\nc{\cB}{{\cal B}}
\nc{\cC}{{\cal C}}
\nc{\cD}{{\cal D}}
\nc{\cE}{{\cal E}}
\nc{\cF}{{\cal F}}
\nc{\cG}{{\cal G}}
\nc{\cH}{{\cal H}}
\nc{\cI}{{\cal I}}
\nc{\cJ}{{\cal J}}
\nc{\cK}{{\cal K}}
\nc{\cL}{{\cal L}}
\nc{\cM}{{\cal M}}
\nc{\cN}{{\cal N}}
\nc{\cO}{{\cal O}}
\nc{\cP}{{\cal P}}
\nc{\cQ}{{\cal Q}}
\nc{\cR}{{\cal R}}
\nc{\cS}{{\cal S}}
\nc{\cT}{{\cal T}}
\nc{\cV}{{\cal V}}
\nc{\cX}{{\cal X}}
\nc{\cY}{{\cal Y}}
\nc{\cZ}{{\cal Z}}
\nc{\cW}{{\cal W}}

\nc{\RR}{{{\mathbb R}}}
\nc{\CC}{{{\mathbb C}}}
\nc{\FF}{{{\mathbb F}}}
\nc{\NN}{{{\mathbb N}}}
\nc{\ZZ}{{{\mathbb Z}}}
\nc{\PP}{{{\mathbb P}}}
\nc{\QQ}{{{\mathbb Q}}}
\nc{\UU}{{{\mathbb U}}}
\nc{\EE}{{{\mathbb E}}}
\nc{\id}{{\operatorname{id}}}

\nc{\1}{{\mathds{1}}}
\nc{\supp}{{\operatorname{supp}}}

\makeatletter
\def\grd@save@target#1{%
  \def\grd@target{#1}}
\def\grd@save@start#1{%
  \def\grd@start{#1}}
\tikzset{
  grid with coordinates/.style={
    to path={%
      \pgfextra{%
        \edef\grd@@target{(\tikztotarget)}%
        \tikz@scan@one@point\grd@save@target\grd@@target\relax
        \edef\grd@@start{(\tikztostart)}%
        \tikz@scan@one@point\grd@save@start\grd@@start\relax
        \draw[minor help lines,magenta] (\tikztostart) grid (\tikztotarget);
        \draw[major help lines] (\tikztostart) grid (\tikztotarget);
        \grd@start
        \pgfmathsetmacro{\grd@xa}{\the\pgf@x/1cm}
        \pgfmathsetmacro{\grd@ya}{\the\pgf@y/1cm}
        \grd@target
        \pgfmathsetmacro{\grd@xb}{\the\pgf@x/1cm}
        \pgfmathsetmacro{\grd@yb}{\the\pgf@y/1cm}
        \pgfmathsetmacro{\grd@xc}{\grd@xa + \pgfkeysvalueof{/tikz/grid with coordinates/major step}}
        \pgfmathsetmacro{\grd@yc}{\grd@ya + \pgfkeysvalueof{/tikz/grid with coordinates/major step}}
        \foreach \x in {\grd@xa,\grd@xc,...,\grd@xb}
        \node[anchor=north] at (\x,\grd@ya) {\pgfmathprintnumber{\x}};
        \foreach \y in {\grd@ya,\grd@yc,...,\grd@yb}
        \node[anchor=east] at (\grd@xa,\y) {\pgfmathprintnumber{\y}};
      }
    }
  },
  minor help lines/.style={
    help lines,
    step=\pgfkeysvalueof{/tikz/grid with coordinates/minor step}
  },
  major help lines/.style={
    help lines,
    line width=\pgfkeysvalueof{/tikz/grid with coordinates/major line width},
    step=\pgfkeysvalueof{/tikz/grid with coordinates/major step}
  },
  grid with coordinates/.cd,
  minor step/.initial=.2,
  major step/.initial=1,
  major line width/.initial=2pt,
}
\makeatother

\makeatletter
\newcommand*\rel@kern[1]{\kern#1\dimexpr\macc@kerna}
\newcommand*\widebar[1]{%
  \begingroup
  \def\mathaccent##1##2{%
    \rel@kern{0.8}%
    \overline{\rel@kern{-0.8}\macc@nucleus\rel@kern{0.2}}%
    \rel@kern{-0.2}%
  }%
  \macc@depth\@ne
  \let\math@bgroup\@empty \let\math@egroup\macc@set@skewchar
  \mathsurround\z@ \frozen@everymath{\mathgroup\macc@group\relax}%
  \macc@set@skewchar\relax
  \let\mathaccentV\macc@nested@a
  \macc@nested@a\relax111{#1}%
  \endgroup
}
\makeatother

%% file: pretex-extra.tex
\usepackage[ruled]{algorithm2e}
\usepackage{pgfplots}
\pgfplotsset{width=8.6cm, compat=1.17} % Set figure size, use compatible version
\usepgfplotslibrary{statistics}
\usepgfplotslibrary{colorbrewer} % For better color schemes
\usepgflibrary{plotmarks} % LATEX and plain TEX and pure pgf
\usepgflibrary[plotmarks] % ConTEXt and pure pgf
\usetikzlibrary{plotmarks} % LATEX and plain TEX when using TikZ
\usetikzlibrary[plotmarks] % ConTEXt when using TikZ
\usepackage{filecontents} % To embed the CSV data directly
\usepackage{calc} % Enable coordinate calculation
\expandafter\def\expandafter\UrlBreaks\expandafter{\UrlBreaks%  save the current one
  \do\a\do\b\do\c\do\d\do\e\do\f\do\g\do\h\do\i\do\j%
  \do\k\do\l\do\m\do\n\do\o\do\p\do\q\do\r\do\s\do\t%
  \do\u\do\v\do\w\do\x\do\y\do\z\do\A\do\B\do\C\do\D%
  \do\E\do\F\do\G\do\H\do\I\do\J\do\K\do\L\do\M\do\N%
  \do\O\do\P\do\Q\do\R\do\S\do\T\do\U\do\V\do\W\do\X%
  \do\Y\do\Z\do\0\do\1\do\2\do\3\do\4\do\5\do\6\do\7\do\8\do\9} % Allow url to breakline after every character